\documentclass[11pt]{article}

\usepackage[a4paper,margin=27mm]{geometry}
\usepackage[T1]{fontenc}
\usepackage{lmodern}
\usepackage{microtype}
\usepackage{amsmath,amssymb,amsthm,mathtools}
\usepackage{booktabs,array}
\usepackage{enumitem}
\usepackage{xcolor}
\usepackage[unicode,colorlinks=true,linkcolor=blue!45!black,
  citecolor=blue!45!black,urlcolor=blue!45!black]{hyperref}
\usepackage[nameinlink,noabbrev]{cleveref}

\newtheorem{theorem}{Theorem}[section]
\newtheorem{proposition}[theorem]{Proposition}
\newtheorem{lemma}[theorem]{Lemma}
\newtheorem{corollary}[theorem]{Corollary}
\theoremstyle{definition}

\newcommand{\F}{\mathbb F}
\newcommand{\K}{\mathbb K}
\newcommand{\rev}{\operatorname{rev}}
\newcommand{\M}{\mathsf M}
\newcommand{\B}{\mathsf B}
\newcommand{\E}{\mathsf E}
\newcommand{\I}{\mathsf I}
\newcommand{\LinComb}{\mathsf{LinComb}}

\title{Generic Characteristic-Zero Equivalence Between\\
Derivative B\'ezout Inversion and Multipoint Evaluation}
\author{Zijian Zeng\\
\small Institute of Computer Science and Digital Innovation,\\
\small UCSI University, Kuala Lumpur, 56000, MALAYSIA\\
\small \href{mailto:zijianzeng@foxmail.com}{zijianzeng@foxmail.com}\quad
\href{mailto:1002266693@ucsiuniversity.edu.my}{1002266693@ucsiuniversity.edu.my}}
\date{August 27, 2026}
\hypersetup{
  pdftitle={Generic Characteristic-Zero Equivalence Between Derivative Bezout Inversion and Multipoint Evaluation},
  pdfauthor={Zijian Zeng},
  pdfsubject={Arithmetic complexity of a canonical polynomial Bezout pair},
  pdfkeywords={Bezout identity, multipoint evaluation, interpolation, algebraic complexity}
}

\begin{document}
\maketitle

\begin{abstract}
Let $a_1,\ldots,a_m$ be distinct elements of a field and let
$Z(X)=\prod_i(X-a_i)$.  We study the total arithmetic complexity of the
canonical B\'ezout pair
\[
  sZ+tZ'=1,\qquad \deg t<m,\quad \deg s<m-1.
\]
The standard product-tree route uses $O(\M_{\F}(m)\log m)$ field
operations.  Even when $\M_{\F}(m)=O(m\log m)$, this is
$O(m\log^2m)$ rather than $O(m\log m)$.  Our main result identifies the
missing logarithm.  Over an infinite field of characteristic zero, in the
generic rational straight-line-program model, computing all coefficients of
$(s,t)$ is equivalent, up to $O(\M_{\F}(m))$ operations, to arbitrary-node
multipoint evaluation and to interpolation.  The new direction rests on an
explicit quasi-linear reconstruction of $Z$ from $(s,t)$ on a nonempty
Zariski-open set.  It also transfers Strassen's
$\Omega(m\log m)$ nonscalar lower bound for elementary symmetric functions
to the B\'ezout problem, showing that the requested order would be optimal.
The reconstruction is genuinely characteristic-dependent: in characteristic
$p$, the squarefree polynomials $X^p+cX+d$ all have the same canonical pair
$(0,c^{-1})$ for fixed $c\ne0$.  Thus the all-field, all-input
$O(m\log m)$ question is not resolved here; it is reduced generically in
characteristic zero to the corresponding arbitrary-node evaluation problem.
\end{abstract}



\section{Introduction}

For distinct nodes $a_1,\ldots,a_m\in\F$, the polynomial
$Z=\prod_i(X-a_i)$ is squarefree.  Hence $Z'$ is invertible in
$\F[X]/(Z)$, and there is a unique normalized pair
\begin{equation}\label{eq:bezout}
  s(X)Z(X)+t(X)Z'(X)=1,
  \qquad \deg t<m,\quad \deg s<m-1.
\end{equation}
The question is whether all $2m-1$ coefficients of this pair can be computed
directly from the nodes using $O(m\log m)$ operations in the ground field.

This question sits at a delicate boundary between two complexity measures.
Borodin and Moenck's classical modular-transform algorithms use
$O(m\log^2m)$ total arithmetic operations but only $O(m\log m)$
multiplications in their model~\cite{BorodinMoenck1974}.  The distinction is
essential: an upper bound on nonscalar multiplications is not an upper bound
on all additions, subtractions, multiplications, and divisions.  Conversely,
the usual soft-linear notation hides exactly the logarithm at issue here.

Our contribution has three parts.
\begin{enumerate}[leftmargin=2.1em]
\item We derive a differential reconstruction formula that recovers the root
  polynomial $Z$ from its canonical pair in $O(\M_{\F}(m))$ operations on a
  nonempty open set in characteristic zero.
\item Combining this formula with automatic differentiation and the
  transposition reductions of Bostan and Schost~\cite{BostanSchost2004}, we
  prove that generic B\'ezout-pair computation, arbitrary-node evaluation,
  and arbitrary-node interpolation have the same complexity up to an
  additive polynomial multiplication.
\item We transfer Strassen's $\Omega(m\log m)$ nonscalar lower bound for the
  elementary symmetric functions~\cite{Strassen1973} to the B\'ezout pair.
  A positive-characteristic family then shows why the reconstruction cannot
  be used without its hypotheses.
\end{enumerate}

The result is deliberately scoped.  It is neither an $O(m\log m)$ algorithm
for general nodes nor a lower bound excluding one.  The standard general
subproduct-tree bound remains $O(\M_{\F}(m)\log m)$, as also stated in a
recent treatment over arbitrary fields~\cite{vanderHoevenLecerf2025};
quasi-linear algorithms are known for important special node
sets~\cite{BostanSchost2005}, not for arbitrary input nodes.  A literature
search through August 27, 2026 found no prior statement of the reconstruction
or the resulting equivalence, but this negative search is not a proof of
novelty.

\section{Models and the classical upper bound}\label{sec:model}

\subsection{Arithmetic and bit models}

An \emph{arithmetic operation} is one addition, subtraction, multiplication,
or division in the stated field.  We write $\M_{\F}(n)$ for an admissible,
superlinear upper bound for multiplying two degree-$<n$ polynomials over
$\F$.  Standard closure assumptions are understood: truncation, balanced
sums, and geometric sequences of problem sizes cost $O(\M_{\F}(n))$.

A \emph{rational straight-line program} (SLP) over an infinite field
$\K$ starts from input coordinates and constants in $\K$ and uses the four
arithmetic operations.  A generic SLP need only be defined on a nonempty
Zariski-open subset.  This is weaker than an algorithm that succeeds for
every tuple of distinct nodes.  We count total SLP size unless the phrase
\emph{nonscalar complexity} is used; in that latter measure, additions and
multiplications by constants are free.

The field-operation model says nothing by itself about bit complexity.  Over
$\mathbb Q$, number fields, or finite fields, a bit bound must also charge for
the representation and growth of field elements, inversion, and writing the
output.  For a fixed finite field $\F_q$, distinctness imposes $m\le q$, so
an asymptotic claim must specify a family of fields or extensions.  In the
prime field $\F_p$, even one field operation has a bit cost depending on
$\log p$ and on the integer multiplication, reduction, and inversion
algorithms.  An in-field FFT of a given length also requires roots of unity
of that length in $\F_p$, hence an order dividing $p-1$; if an extension is
used instead, its construction and arithmetic costs must be charged.

Nor do we assume FFT-friendly roots of unity.  If the ground field supplies
suitable smooth-order roots at unit arithmetic cost, one may have
$\M_{\F}(m)=O(m\log m)$.  Over arbitrary algebras, Cantor and Kaltofen give
$O(m\log m)$ algebra multiplications together with
$O(m\log m\log\log m)$ additions and subtractions
\cite{CantorKaltofen1991}.  Thus a valid total-operation consequence is
$\M_{\F}(m)=O(m\log m\log\log m)$, not $O(m\log m)$.

\subsection{Canonical form and interpolation}

\begin{lemma}[Canonical pair]\label{lem:canonical}
For pairwise distinct $a_i$, there is a unique pair satisfying
\eqref{eq:bezout}.  It obeys
\begin{equation}\label{eq:values}
  t(a_i)=\frac1{Z'(a_i)},
  \qquad
  t(X)=\sum_{i=1}^m\frac{1}{Z'(a_i)^2}\frac{Z(X)}{X-a_i}.
\end{equation}
\end{lemma}

\begin{proof}
Squarefreeness makes $Z'$ a unit modulo $Z$.  Its inverse has a unique
representative $t$ of degree $<m$, and then
$s=(1-tZ')/Z$ is unique and has degree at most $m-2$.  Evaluating
\eqref{eq:bezout} at $a_i$ gives the first identity in
\eqref{eq:values}.  Interpolation in the Lagrange basis
$Z/((X-a_i)Z'(a_i))$ gives the second.
\end{proof}

Formula~\eqref{eq:values} exposes the standard algorithm.  A product tree
constructs $Z$; a subproduct remainder tree evaluates $Z'$ at the nodes; the
$m$ nonzero values are inverted; interpolation constructs $t$; and one exact
division gives $s$.  Polynomial differentiation is linear-time, while each
tree level costs $O(\M_{\F}(m))$.  Consequently
\begin{equation}\label{eq:baseline}
  \text{B\'ezout-pair cost}=O(\M_{\F}(m)\log m)
\end{equation}
over every field.  Fast extended gcd or modular inversion gives the same
asymptotic baseline.  If $\M_{\F}(m)=O(m\log m)$, then
\eqref{eq:baseline} is $O(m\log^2m)$, not the requested bound.  With the
arbitrary-algebra bound above it is
$O(m\log^2m\log\log m)$.

\section{Recovering the root polynomial}\label{sec:reconstruction}

For a polynomial $f$ with degree cap $d$, set
$\rev_d(f)=x^d f(1/x)$.  Define
\begin{equation}\label{eq:reversals}
 A=\rev_m(Z),\qquad T=\rev_{m-1}(t),\qquad S=\rev_{m-2}(s).
\end{equation}
In particular $A(0)=1$.

\begin{theorem}[Differential reconstruction]\label{thm:reconstruction}
Let $m\ge3$ and let $\K$ have characteristic zero.  Put
\[
  \tau=T(0)=[X^{m-1}]t.
\]
On the locus $\tau\ne0$, the canonical pair determines $Z$ by
\begin{equation}\label{eq:reconstruction}
 A=\exp\!\left(\int \frac{C}{T}\,dx\right)\pmod{x^{m+1}},
 \qquad C=\frac{mT+S}{x}.
\end{equation}
The right side is well defined, and all coefficients of $Z$ are recovered in
$O(\M_{\K}(m))$ arithmetic operations.
\end{theorem}

\begin{proof}
Differentiating $Z(1/x)=x^{-m}A(x)$ gives
\[
  Z'(1/x)=x^{1-m}\bigl(mA-xA'\bigr).
\]
Reverse \eqref{eq:bezout} using the degree caps in
\eqref{eq:reversals}.  After multiplication by $x^{2m-2}$, the result is
\begin{equation}\label{eq:reversebezout}
 T(mA-xA')+SA=x^{2m-2}.
\end{equation}
The constant term of the left side is $mT(0)+S(0)$, whereas the right side
has zero constant term.  Hence $mT+S$ is divisible by $x$, so $C$ in
\eqref{eq:reconstruction} is a polynomial.  Dividing
\eqref{eq:reversebezout} by $x$ gives
\[
  CA-TA'=x^{2m-3}.
\]
Because $2m-3\ge m$, reduction modulo $x^m$ yields
\begin{equation}\label{eq:logder}
  \frac{A'}A=\frac CT\pmod{x^m}.
\end{equation}
On $\tau\ne0$, $T$ is a unit in $\K[[x]]$; $A$ is a unit because $A(0)=1$.
Integration in characteristic zero and the normalization $A(0)=1$ turn
\eqref{eq:logder} into \eqref{eq:reconstruction}.  Truncated inversion,
integration, and exponential each cost $O(\M_{\K}(m))$ by the usual Newton
algorithms for power series~\cite{BrentKung1978,vonZurGathenGerhard2013}.
Reversing $A$ recovers $Z$.
\end{proof}

\begin{lemma}\label{lem:open}
The condition $\tau\ne0$ defines a nonempty Zariski-open subset of the
distinct-node locus.
\end{lemma}

\begin{proof}
The leading coefficient in \eqref{eq:values} is
\begin{equation}\label{eq:tau}
  \tau=\sum_{i=1}^m\frac1{Z'(a_i)^2}.
\end{equation}
For any distinct rational real nodes every summand is positive.  Therefore
the rational function in \eqref{eq:tau} is not identically zero.
\end{proof}

The open-set qualification cannot simply be removed.  For the nodes
$0,1,r$ one obtains
\[
 \tau=\frac{2(r^2-r+1)}{r^2(r-1)^2},
\]
which vanishes over a characteristic-zero field containing a root of
$r^2-r+1$.  The small cases $m=1,2$ are direct and do not affect asymptotic
generic complexity.

\section{Equivalence with evaluation and interpolation}\label{sec:equiv}

Fix an infinite characteristic-zero field $\K$.  Let $\B(m)$ be the minimum
generic rational-SLP size for computing all coefficients of the canonical
pair from arbitrary input nodes.  Let $\E(m)$ be the minimum size of a
generic SLP for
\[
  (a_1,\ldots,a_m;p_0,\ldots,p_{m-1})
   \longmapsto (P(a_1),\ldots,P(a_m)),
  \quad P=\sum_{j<m}p_jX^j,
\]
that is linear in the coefficient vector $p$.  Define $\I(m)$ analogously for
interpolation, with the value vector as the linear input.  The nodes are part
of the input: no preprocessing of a fixed special set is allowed.

\begin{theorem}[Generic equivalence]\label{thm:equivalence}
In the model above,
\begin{equation}\label{eq:equivalence}
 \B(m)+\M_{\K}(m)
 =\Theta\bigl(\E(m)+\M_{\K}(m)\bigr)
 =\Theta\bigl(\I(m)+\M_{\K}(m)\bigr).
\end{equation}
\end{theorem}

\begin{proof}
Bostan and Schost prove in characteristic zero that evaluation and
interpolation are equivalent up to a constant number of polynomial
multiplications~\cite[Theorems 1 and 2]{BostanSchost2004}:
\begin{equation}\label{eq:EI}
 \E(m)=O(\I(m)+\M_{\K}(m)),\qquad
 \I(m)=O(\E(m)+\M_{\K}(m)).
\end{equation}

First suppose an interpolation program is available.  At the nodes, the
value vector
\[
 \left(\prod_{j=2}^m(a_1-a_j),0,\ldots,0\right)
\]
interpolates to $\prod_{j=2}^m(X-a_j)$.  One multiplication by $X-a_1$
therefore gives $Z$.  Use the evaluation program furnished by
\eqref{eq:EI} to evaluate $Z'$ at all nodes, invert the $m$ nonzero values,
and interpolate once more to obtain $t$ from \eqref{eq:values}.  Finally,
$s=(1-tZ')/Z$ is an exact quotient.  Fast exact division costs
$O(\M_{\K}(m))$.  Hence
\begin{equation}\label{eq:Bupper}
 \B(m)=O(\I(m)+\M_{\K}(m))
      =O(\E(m)+\M_{\K}(m)).
\end{equation}

For the converse, compose a size-$\B(m)$ pair program with
\cref{thm:reconstruction}.  This computes the coefficients of $Z$ on a
nonempty open set in size $O(\B(m)+\M_{\K}(m))$.  Give each node an
independent tangent $c_i$ and forward-different every SLP instruction.
The product and quotient rules have constant SLP size, and every tangent is
linear in $c=(c_1,\ldots,c_m)$.  The output tangent polynomial is
\begin{equation}\label{eq:lincomb}
 \dot Z
 =\sum_{i=1}^m c_i\frac{\partial Z}{\partial a_i}
 =-\sum_{i=1}^m c_i\prod_{j\ne i}(X-a_j).
\end{equation}
Thus we obtain, up to sign, the $\LinComb$ map of
Bostan--Schost.  Their transposition reduction from $\LinComb$ to arbitrary-
node evaluation adds $O(\M_{\K}(m))$ operations
\cite[Section 6]{BostanSchost2004}.  It follows that
\begin{equation}\label{eq:Eupper}
 \E(m)=O(\B(m)+\M_{\K}(m)).
\end{equation}
Combining \eqref{eq:EI}, \eqref{eq:Bupper}, and \eqref{eq:Eupper} proves
\eqref{eq:equivalence}.
\end{proof}

\begin{corollary}\label{cor:target}
Whenever the multiplication regime satisfies
$\M_{\K}(m)=O(m\log m)$, generic B\'ezout-pair computation has an
$O(m\log m)$ total-operation algorithm if and only if arbitrary-node
multipoint evaluation does, and if and only if arbitrary-node interpolation
does.
\end{corollary}

This is an equivalence, not an assertion that any of the three target-time
algorithms is known.  It also concerns generic programs; the exceptional
example following Lemma~\ref{lem:open} prevents an immediate all-input conclusion.

\section{Optimality in nonscalar complexity}\label{sec:lower}

For clarity, let $L_{\K}(f_1,\ldots,f_r)$ denote Strassen's measure: the
minimum number of multiplications and divisions in a branch-free computation
of rational functions over an infinite ground field $\K$, with additions,
subtractions, and multiplication by elements of $\K$ free.  Such a
computation is an identity in a rational function field and therefore need
only be defined on a nonempty Zariski-open set.  Strassen's Theorem~3.1 gives,
for the elementary symmetric functions $e_1,\ldots,e_m$,
\begin{equation}\label{eq:strassen}
 L_{\K}(e_1,\ldots,e_m)\ge m\log_2(m/e)
 \qquad (\K\ \text{infinite}).
\end{equation}
See~\cite{Strassen1973}.  This is precisely a nonscalar rational-SLP lower
bound, not a lower bound for branching algorithms or finite-domain models.
Since the coefficients of $Z$ are the $e_i$ up to signs, the reconstruction
transfers the bound within the generic model.

\begin{proposition}\label{prop:lower}
Every generic rational SLP that computes all coefficients of the canonical
pair over a characteristic-zero field has nonscalar complexity
$\Omega(m\log m)$.  Consequently its total arithmetic complexity is also
$\Omega(m\log m)$.
\end{proposition}

\begin{proof}
Fix a characteristic-zero field $\K$.  For every precision $n$, it contains
$2n-1$ distinct constants.  Multiplication of two degree-$<n$ polynomials can
therefore be performed by evaluation at these fixed constants, $2n-1$
pointwise products, and interpolation.  The two transforms are
$\K$-linear, hence free in Strassen's measure; the multiplication uses at
most $2n-1$ nonscalar operations.

Choose standard doubling Newton circuits for truncated inverse and
exponential.  There is an absolute constant $q$ such that each doubling to
precision $n$ uses at most $q$ truncated polynomial products, together with
linear operations and a bounded number of nonscalar divisions.  Replacing
each product by the fixed-node circuit above and summing the geometric
precisions shows that the entire reconstruction uses at most $Cm$
nonscalar operations, for one constant $C$ independent of $m$.  Reversal,
differentiation, division by $x$, and integration are linear; the divisions
by integration indices are scalar operations and are legal in characteristic
zero.

Let $L_{\mathrm{pair},\K}(m)$ be the minimum nonscalar size of a generic
rational SLP for the canonical pair.  Composing such a program with the
reconstruction computes $e_1,\ldots,e_m$, so \eqref{eq:strassen} gives the
pointwise inequality
\[
 L_{\mathrm{pair},\K}(m)+Cm\ge m\log_2(m/e).
\]
Consequently
\[
 L_{\mathrm{pair},\K}(m)
 \ge m\log_2m-(C+\log_2e)m
 \ge \tfrac12m\log_2m
\]
for all sufficiently large $m$.  This proves the claimed
$\Omega(m\log m)$ nonscalar lower bound.  Since every nonscalar operation is
a field operation, the same lower bound holds for total rational-SLP size.
\end{proof}

Within this generic branch-free model, the lower bound matches the order
asked for; it proves optimality if the target algorithm exists, but it does
not exclude such an algorithm.  No branching, finite-domain, or
positive-characteristic lower bound is asserted.

\section{Positive characteristic and finite fields}\label{sec:positive}

The characteristic-zero hypothesis in the reconstruction is structural, not
merely an artifact of division by the integration indices.

\begin{proposition}[Failure of reconstruction]\label{prop:positive}
Let $p>0$ be prime and work over an algebraically closed field of
characteristic $p$.  For fixed $c\ne0$, every polynomial
\[
  Z_d(X)=X^p+cX+d
\]
is squarefree and splits into $p$ distinct linear factors, but its canonical
pair is $(s,t)=(0,c^{-1})$, independently of $d$.
\end{proposition}

\begin{proof}
$Z_d'=c$, so $Z_d$ is squarefree.  The identity
$0\cdot Z_d+c^{-1}Z_d'=1$ already satisfies the degree normalization, hence
is the unique canonical pair.
\end{proof}

Thus $(s,t)$ need not determine $Z$ in positive characteristic, and the hard
direction of \cref{thm:equivalence} does not extend verbatim.  A finite-field
edge case is even more explicit: if $m=q$ and the nodes are all of $\F_q$,
then $Z=X^q-X$ and the pair is $(0,-1)$.  This does not settle the complexity
for growing families with $m<q$; it shows that generic infinite-field lower
bounds cannot be imported into a field-specific finite-domain model without
new quantifiers and arguments.

\section{Complexity consequences and unresolved scope}\label{sec:conclusion}

The original all-field, all-distinct-input question therefore remains open.
The rigorous endpoint is a partial theorem: in characteristic zero and on a
nonempty Zariski-open set, the B\'ezout problem is equivalent to the general
arbitrary-node evaluation/interpolation problem, and in the same generic
branch-free rational-SLP model it has the matching $\Omega(m\log m)$ lower
bound.  Any
complete resolution must either remove the generic qualifier, find the
missing-logarithm evaluation algorithm, or use positive-characteristic
structure that survives
Proposition~\ref{prop:positive}.

\section*{Reproducibility and disclosure}

The algebraic identities were checked on exact rational instances and on
symbolic low-degree cases by two independent scripts included with the
supplementary source archive.  These checks are diagnostic and are not used
as proofs.  Language-model assistance was used for literature triage,
symbolic diagnostics, and editorial revision; every theorem in this note is
supported by the displayed arguments and cited results.  No empirical data
were collected.

For reference, \Cref{tab:bounds} summarizes the principal models.  Every row
counts total field operations unless explicitly marked otherwise.

\begin{table}[ht]
\centering
\small
\begin{tabular}{@{}>{\raggedright\arraybackslash}p{0.27\textwidth}
                    >{\raggedright\arraybackslash}p{0.28\textwidth}
                    >{\raggedright\arraybackslash}p{0.35\textwidth}@{}}
\toprule
Setting & Proven general bound & Consequence for the question \\
\midrule
Arbitrary field, arbitrary nodes
  & $O(\M_{\F}(m)\log m)$
  & Classical product/subproduct-tree upper bound. \\
FFT-friendly multiplication
  & $O(m\log^2m)$
  & Does not give $O(m\log m)$. \\
Arbitrary-algebra multiplication
  & $O(m\log^2m\log\log m)$
  & Counts the Cantor--Kaltofen additions honestly. \\
Characteristic zero, generic SLP
  & Equivalent to evaluation/interpolation up to $O(\M(m))$
  & The $O(m\log m)$ target is exactly the corresponding arbitrary-node target when $\M(m)=O(m\log m)$. \\
Characteristic zero, generic branch-free rational SLP
  & $\Omega(m\log m)$
  & The requested order, if attained, is optimal. \\
Special node families
  & Often $O(\M(m))$
  & Does not cover arbitrary distinct input nodes. \\
\bottomrule
\end{tabular}
\caption{Arithmetic-complexity conclusions.  No bit bound follows without a
field representation and coefficient-height analysis.}
\label{tab:bounds}
\end{table}

\bibliographystyle{abbrv}
\bibliography{references}

@article{BorodinMoenck1974,
  author  = {Allan Borodin and Robert Moenck},
  title   = {Fast Modular Transforms},
  journal = {Journal of Computer and System Sciences},
  volume  = {8},
  number  = {3},
  pages   = {366--386},
  year    = {1974},
  doi     = {10.1016/S0022-0000(74)80029-2}
}

@article{BostanSchost2004,
  author  = {Alin Bostan and {\'E}ric Schost},
  title   = {On the Complexities of Multipoint Evaluation and Interpolation},
  journal = {Theoretical Computer Science},
  volume  = {329},
  number  = {1--3},
  pages   = {223--235},
  year    = {2004},
  doi     = {10.1016/j.tcs.2004.09.002}
}

@article{BostanSchost2005,
  author  = {Alin Bostan and {\'E}ric Schost},
  title   = {Polynomial Evaluation and Interpolation on Special Sets of Points},
  journal = {Journal of Complexity},
  volume  = {21},
  number  = {4},
  pages   = {420--446},
  year    = {2005},
  doi     = {10.1016/j.jco.2004.09.009}
}

@article{BrentKung1978,
  author  = {Richard P. Brent and H. T. Kung},
  title   = {Fast Algorithms for Manipulating Formal Power Series},
  journal = {Journal of the ACM},
  volume  = {25},
  number  = {4},
  pages   = {581--595},
  year    = {1978},
  doi     = {10.1145/322092.322099}
}

@article{CantorKaltofen1991,
  author  = {David G. Cantor and Erich Kaltofen},
  title   = {On Fast Multiplication of Polynomials over Arbitrary Algebras},
  journal = {Acta Informatica},
  volume  = {28},
  number  = {7},
  pages   = {693--701},
  year    = {1991},
  doi     = {10.1007/BF01178683}
}

@article{Strassen1973,
  author  = {Volker Strassen},
  title   = {Die Berechnungskomplexit{\"a}t von elementarsymmetrischen
             Funktionen und von Interpolationskoeffizienten},
  journal = {Numerische Mathematik},
  volume  = {20},
  pages   = {238--251},
  year    = {1973},
  doi     = {10.1007/BF01436566}
}

@article{vanderHoevenLecerf2025,
  author  = {Joris van der Hoeven and Gr{\'e}goire Lecerf},
  title   = {Faster Multi-Point Evaluation over Any Field},
  journal = {Applicable Algebra in Engineering, Communication and Computing},
  year    = {2025},
  note    = {Published online},
  doi     = {10.1007/s00200-025-00704-7}
}

@book{vonZurGathenGerhard2013,
  author    = {Joachim von zur Gathen and J{\"u}rgen Gerhard},
  title     = {Modern Computer Algebra},
  edition   = {3},
  publisher = {Cambridge University Press},
  address   = {Cambridge},
  year      = {2013},
  doi       = {10.1017/CBO9781139856065}
}

\end{document}